\documentclass[11pt]{article}
\usepackage[latin9]{inputenc}
\usepackage[letterpaper]{geometry}
\usepackage[active]{srcltx}
\usepackage{color}
\usepackage{verbatim}
\usepackage{float}
\usepackage{amsthm}
\usepackage{amsmath}
\usepackage{amssymb}
\usepackage{graphicx}
\usepackage{setspace}
\PassOptionsToPackage{normalem}{ulem}
\usepackage{ulem}
\definecolor{darkblue}{rgb}{0.0,0.0,0.5}
\usepackage[hidelinks,colorlinks=true,citecolor=darkblue,linkcolor=black]{hyperref}
\usepackage{breakurl}

\makeatletter

\floatstyle{ruled}
\newfloat{algorithm}{tbp}{loa}
\providecommand{\algorithmname}{Algorithm}
\floatname{algorithm}{\protect\algorithmname}

\usepackage{enumitem}		% customizable list environments
\theoremstyle{plain}
\newtheorem{thm}{\protect\theoremname}
  \theoremstyle{plain}
  \newtheorem{lem}{\protect\lemmaname}
  \theoremstyle{definition}
  \newtheorem{defn}{\protect\definitionname}

\@ifundefined{date}{}{\date{}}
\usepackage{color}
\usepackage{mathrsfs}
\usepackage{amsthm}
\usepackage{breakurl}
\usepackage{perpage}
\usepackage{epstopdf}

\MakePerPage{footnote}

\usepackage{sober}
\setlist[description]{leftmargin=0pt}

\floatname{algorithm}{Protocol}

\DeclareMathAlphabet{\mathpzc}{OT1}{pzc}{m}{it}

\let\oldenumerate=\enumerate
\def\enumerate{
\oldenumerate
\setlength{\itemsep}{2pt}
\vspace{-4pt}
}

\let\olditemize=\itemize
\def\itemize{
\olditemize
\setlength{\itemsep}{2pt}
\vspace{-3pt}
}

\makeatother

  \providecommand{\definitionname}{Definition}
  \providecommand{\lemmaname}{Lemma}
\providecommand{\theoremname}{Theorem}

\begin{document}
\global\long\def\bla{\big\langle}
 \global\long\def\bra{\big\rangle}
 \global\long\def\zp{\mathbb{Z}_{p}}
 \global\long\def\zq{\zp}
 \global\long\def\deg{\mathsf{deg}}

\global\long\def\la{\langle}
 \global\long\def\ra{\rangle}

\title{Secure Anonymous Broadcast \bigskip{}
}

\author{Mahnush Movahedi\\
 movahedi@cs.unm.edu \and Jared Saia\\
 saia@cs.unm.edu \and Mahdi Zamani\\
 zamani@cs.unm.edu}

\date{{\small{Department of Computer Science, University of New Mexico,
Albuquerque, NM, USA 87131}}}
\maketitle
\begin{abstract}
In anonymous broadcast, one or more parties want to anonymously send
messages to all parties. This problem is increasingly important as
a black-box in many privacy-preserving applications such as anonymous
communication, distributed auctions, and multi-party computation.
In this paper, we design decentralized protocols for anonymous broadcast
that require each party to send (and compute) a polylogarithmic number
of bits (and operations) per anonymous bit delivered with $O(\log n)$
rounds of communication. Our protocol is provably secure against traffic
analysis, does not require any trusted party, and is completely load-balanced.
The protocol tolerates up to $n/6$ statically-scheduled Byzantine
parties that are controlled by a computationally unbounded adversary.
Our main strategy for achieving scalability is to perform local communications
(and computations) among a logarithmic number of parties. We provide
simulation results to show that our protocol improves significantly
over previous work. We finally show that using a common cryptographic
tool in our protocol one can achieve practical results for anonymous
broadcast.
\end{abstract}
\medskip{}

\section{Introduction}

\setcounter{page}{1}

\begin{comment}
Anonymous communication allows individuals to communicate with each
other without fear of surveillance. An anonymity system attempts to
conceal the mapping between messages and their intended recipients,
between messages and their actual senders, or both (\emph{full anonymity}).
\end{comment}
{} Today, political and commercial entities are increasingly engaging
in sophisticated cyber-warfare to damage, disrupt, or censor information
content~\cite{cyberwarfare2012}. In designing anonymous communication
services, there is a need to ensure reliability even against very
powerful types of adversaries. Such adversaries can monitor large
portions of networks and control a certain fraction of the parties
to run sophisticated \emph{active} attacks such as jamming, corruption,
and forging, as well as simple \emph{passive} attacks such as eavesdropping
and non-participation.%
\begin{comment}
So far, most anonymity schemes have assumed the existence of a trusted
third party. Unfortunately, finding such trusted parties in practice
is usually very hard or even impossible. Due to the huge number of
users interested in anonymously sharing their data frequently, such
centralized parties should be powerful servers that are usually owned
by large companies and governments. Moreover, centralized parties
may be subject to governmental control, and may be banned or forced
to disclose sensitive information.
\end{comment}

In this paper, we focus on the problem of secure anonymous broadcast,
where a set of $n$ parties want to anonymously send their messages
to all parties. Anonymous broadcast is an important tool for achieving
privacy in several distributed applications such as anonymous communication~\cite{chaum:81},
private information retrieval~\cite{Chor:1998:PIR:293347.293350},
secure auctions~\cite{stajano:99}, and multi-party computation (MPC).
Our goal is to design a decentralized anonymous broadcast protocol
that scales well with the number of parties and is robust against
an active adversary. One motivating application for this protocol
is a decentralized version of Twitter that enables provably-anonymous
broadcast of messages.

One challenging problem with most anonymity-based systems is resistance
against \emph{traffic-analysis}. A global adversary can monitor traffic
exchanged between parties to link messages to the their corresponding
senders. Such a powerful adversary was assumed to be unrealistic in
the past but it is believed to be realistic today especially if the
service provider is controlled or compromised by a state-level surveillance
authority~\cite{Feigenbaum:2014:Panopticon}. Unfortunately, well-known
anonymous services such as Crowds~\cite{Reiter:1998:CAW:290163.290168}
and Tor~\cite{dingledine:2004} are not secure against traffic analysis
attacks.%
\begin{comment}
If the adversary monitors all communications it is necessary that
all parties send a message in every round of the protocol. While this
sometimes puts a large burden on the network, we believe with the
fast growth of the number of users
\end{comment}

Two widely-accepted architectures for providing general anonymity
against an active adversary are \emph{Mix networks (Mix-Nets)} and
\emph{Dining Cryptographers networks (DC-Nets)}, both of which were
originally proposed by Chaum~\cite{chaum:81,chaum88}. Mix-Nets require
semi-trusted infrastructure nodes and are known to be vulnerable to
traffic analysis and active attacks~\cite{pfitzmann:86}. DC-Nets~\cite{chaum88,Juels04,vonAhn03,Waidner:1990:DCD:111563.111630},
on the other hand, provide anonymous broadcast protocols among a group
of parties without requiring trusted parties. The core idea of DC-Nets
is that a protocol for multi-party computation can be used to perform
sender and receiver anonymous broadcast. For example, if party $p_{i}$
wants to broadcast a message $m_{i}$ anonymously, then all other
parties participate in a multi-party sum with input zero, while party
$p_{i}$ participates with input $m_{i}$. All parties learn the sum,
which is $m_{i}$ while all inputs remain private. This ensures that
no party can trace the output message $m_{i}$ to its input, keeping
$p_{i}$ anonymous. 

Although DC-Nets are provably-secure against traffic analysis, they
face several challenges. First, a reservation mechanism is required
to schedule which party is broadcasting without compromising the anonymity
of the sender. Second, DC-Nets are susceptible to collisions, which
degrade throughput. A jamming adversary may even use collisions to
render the channel useless by continuously transmitting in every round.
Third, typical DC-Nets are not scalable given that the bit complexity
required to anonymously broadcast a single bit among $n$ parties
is $\Omega(n^{2})$.

State-of-the-art approaches that address some of these challenges
include~\cite{Verdict:2013:usenix,Juels04,vonAhn03}. The majority
of these methods scale poorly with network size, rendering them impractical
for large networks. Recently, Zamani et al.~\cite{Zamani:2013:FOCI,Karlin:2013:BAS:2484239.2484290}
proposed the first anonymous broadcast protocol where each party sends
$o(n)$ bits to broadcast a bit among $n$ parties. Their protocol
uses multi-party computation to achieve full anonymity and logarithmic-size
groups of parties to achieve $\tilde{O}(1)$ communication and computation
costs. Unfortunately, their protocol has polylogarithmic rounds of
communication and is not practical due to large logarithmic factors
hidden in the complexity notation.

To the best of our knowledge, every sender and receiver anonymous
broadcast protocol that does not rely on a trusted party consists
of at least three steps.\vspace{1pt}

\begin{enumerate}
\item \textbf{Input: }Initially, each party holds a message. The party distributes
its message or a representation of it among all or a subset of parties.
This step requires sending $\Omega(n)$ messages.
\item \textbf{Multi-party shuffling: }All or a subset of parties participate
in a multi-party protocol to obliviously generate a random permutation
of the sequence of message they hold.
\item \textbf{Output:} All or a subset of parties holding a sequence of
messages broadcast them to all parties. This step requires sending
$\Omega(n^{2})$ messages for delivering $n$ shuffled messages.%
\begin{comment}
Distinguishing between these steps helps us better compare different
anonymous broadcast protocols and identify future directions for each
steps. Moreover, one can think of protocols that invoke the anonymization
step more than once previously broadcast inputs.
\end{comment}

\end{enumerate}
Since much cannot be done to improve the cost of the output phase,
we will mainly focus on the multi-party shuffling step in this paper.
Multi-party shuffling can be used as a black-box in multi-party computation
problems. Boyle~et~al.~\cite{Boyle:2013:CLS:2450206.2450227} use
oblivious shuffling to randomly choose inputs for a sublinear function
evaluation, where a function is evaluated over $o(n)$ inputs chosen
uniformly at random in the presence of an active adversary. Laur~et~al.~\cite{Laur:2011:ObliviousDB}
and Goodrich~et~al.~\cite{Goodrich:2012:POS:2133601.2133604} describe
how multi-party shuffling can be used for implementing oblivious database
operations and oblivious storage.
\begin{description}
\item [{Our~Model.}] In this paper, we consider a network of $n$ parties
whose identities are common knowledge. We assume there is a private
and authenticated communication channel between every pair of parties
and the communication is \emph{synchronous}. Our protocol does not
require the presence of any trusted third party, and we do not assume
the existence of a reliable broadcast channel. We assume $t<(1/6-\epsilon)n$
of the parties are controlled by an \emph{active} adversary, for some
positive constant $\epsilon$. We assume our adversary is \emph{computationally
unbounded} and is actively trying to prevent the protocol from succeeding
by attacking the privacy of the parties, and the integrity of communications,
by attempting to corrupt, forge, or drop messages. We say that the
parties controlled by the adversary are \emph{dishonest} and that
the remaining parties are \emph{honest} meaning that they strictly
follow our protocol. We finally assume that the adversary is \emph{static}
meaning that it must select the set of dishonest parties at the start
of the protocol.
\end{description}

\subsection{Our Contribution}

In this paper, we design a decentralized protocol for anonymous communication
that is secure against an active adversary. Our protocol is efficient
and scales well with the number of parties. Moreover, our protocol
is load-balanced meaning that each party handles a roughly equal amount
of communication and computation. We use techniques from multi-party
computation, where a set of $n$ parties, each having a secret value,
compute a known function over their inputs, without revealing the
inputs to any party.

Recently, Boyle{\small{~}}et~al.~\cite{Boyle:2013:CLS:2450206.2450227}
and Dani~et~al.~\cite{DKMS-ICDCN-2014} proposed scalable solutions
to general MPC. Unfortunately, both of these protocols are not practical
due to large logarithmic and constant factors in their communication/computation
costs. Moreover, the protocol of Boyle~et al. is not load-balanced
making it hard to be used in settings like mobile networks where the
parties have limited resources. Despite their inefficiency, we are
inspired by \cite{Boyle:2013:CLS:2450206.2450227} and \cite{DKMS-ICDCN-2014}
to achieve scalability by performing local communications and computations
in logarithmic-size groups of parties called \emph{quorum}s, where
the number of dishonest parties in each quorum is guaranteed not to
exceed a certain fraction. Using quorums and by simplifying much of
their work, we develop an efficient multi-party shuffling protocol
for anonymizing user inputs.

Our protocol is provably-secure as it is based on a formal security
framework, which follows from the security of MPC. We show that the
anonymity achieved by this method is, in particular, resistant to
traffic analysis. We also provide provable anonymity against \emph{a
priori} knowledge that an adversary might have regarding the potential
communicating parties. Moreover, unlike the majority of previous work
which rely on centralized trusted servers, our protocol is fully-decentralized
and does not require any trusted party.

Our protocol has polylogarithmic communication and computation costs
with respect to the number of parties. We prove the following main
theorem in Section~\ref{sec:Security-Proofs}.
\begin{thm}
\label{thm:main} Consider $n$ parties in a fully connected network
with private channels where each party has a message to send to all
parties. There exists an unconditionally-secure $n$-party protocol
such that if all honest parties follow the protocol, then with high
probability\hspace{0.003\textwidth}%
\footnote{An event occurs \emph{with high probability}, if it occurs with probability
$\geq1-1/n^{c}$, for $c>0$ and sufficiently large $n$.%
}:\vspace{1pt}

\begin{itemize}
\item Each honest party sends its message to all parties anonymously.
\item The protocol tolerates up to $t<(1/6-\epsilon)n$ malicious parties,
for some positive constant $\epsilon$.
\item Each party sends $\tilde{O}(1)$ bits and computes $\tilde{O}(1)$
operations for shuffling $n$ messages\hspace{0.003\textwidth}%
\footnote{The $\tilde{O}$ notation is used as a variant of the big-O notation
that ignores logarithmic factors. Thus, $f(n)=\tilde{O}(g(n))$ means
$f(n)=O(g(n)\log^{k}{g(n)})$ for some $k$.%
}.
\begin{spacing}{0.69999999999999996}
\item Each party sends $\tilde{O}(1)$ bits and computes $\tilde{O}(1)$
operations for delivering one anonymous bit.\smallskip{}

\item The protocol has $O(\log n)$ rounds of communication.\end{spacing}

\end{itemize}
\end{thm}
\begin{description}
\item [{Protocol~Overview.}] In our protocol, $n$ parties participate
in a MPC to jointly compute a shuffling function that randomly permutes
their messages. Then, the results are broadcast to all participating
parties. More specifically, our protocol builds a set of quorums in
a one-time setup phase and then, uses the quorums in the online phase
for shuffling input messages. We represent the desired shuffling function
by an arithmetic circuit, where the computation of each gate is assigned
to a quorum. Then, we evaluate the circuit level-by-level, passing
the outputs of one level as the inputs to the next level. Once the
local computation is finished in each quorum, the result is forwarded
to the next quorum via one-to-one communication with parties of the
next quorum. Finally, at the highest level, the shuffled messages
are computed and sent to all parties.
\end{description}
One technique for randomly shuffling a set of messages in a multi-party
setting is to assign to each message a uniform random value and then,
obliviously sort the messages according to the random values. One
issue with this technique is that if the randoms are not distinct
(i.e., there is a \emph{collision}), then the resulting distribution
deviates from the uniform distribution. Unfortunately, this gives
the adversary some advantage to map the inputs to outputs and thus,
break the anonymity. More formally, let $n$ be the number of messages
and $[1,M]$ be the range of random values we choose for each message.
With $n$ such random values, we can generate $M^{n}$ possible states
while there are $n!$ possible permutations of the messages. Since
$n!\nmid M^{n}$, it is easy to see that some permutations are more
likely to be generated than others. On the other hand, we will show
that by choosing a sufficiently large $M$, we can prevent collisions
with high probability. Using the Chernoff bound, we will prove that
choosing $M=\Omega(kn^{2}\log n)$ guarantees a uniform random shuffle
with probability $1-1/n^{k}$, for any constant $k>0$.

\section{Further Related Work \label{sec:relatedWork}}

Some protocols are built upon a relaxed notion of anonymity called
\emph{$k$-anonymity}\ \cite{vonAhn03,Yao:2005:Anonymity,LeBlond:2013:TET:2486001.2486002},
where the adversary is assumed to be unable to identify the actual
sender/receiver of a message from a set of $k$ parties (called \emph{anonymity
set}). Even though \emph{$k$-anonymity} often increases efficiency
significantly, choosing small $k$'s can result in severe privacy
problems. For example, attackers often have background knowledge and
it is shown that small anonymity sets are likely to leak privacy when
attackers have such knowledge~\cite{Machanavajjhala:2007:LDP:1217299.1217302}.
For example, a person located in New Mexico is more likely to search
for a restaurant serving chili stew than a person in Vermont.

Von Ahn et al.~\cite{vonAhn03} develop a cryptographic broadcast
protocol based on DC-Nets that is resistant to a static active adversary.
A set of $n$ parties with private inputs compute and share the sum
of their inputs without revealing any parties' input. The authors
introduce $k$-anonymity, which means no polynomial-time adversary
may distinguish the sender/receiver of a message from among $k$ honest
senders/receivers. To achieve $k$-anonymity, they partition the set
of parties into groups of size $M=O(k)$ and execute a multi-party
sum protocol inside each group. The jamming detection mechanism is
weak against an adversary who may waste valuable resources by adaptively
filling up to $M$ channels. In the case where $n$-anonymity is desired,
the protocol requires $O(n^{3})$ messages to be sent per anonymous
message and the total bit complexity is $O(n^{4})$. The protocol
has latency that is $O(1)$ on average when the number of broadcasts
is large, but which can be $O(n)$ in worst case for a single broadcast.

Golle and Juels~\cite{Juels04} employ cryptographic proofs of correctness
to solve the jamming problem in DC-Nets assuming a static Byzantine
adversary. The protocol detects jamming with high probability in $O(1)$
rounds, requiring a total communication and computation complexity
of $O(n^{2})$ bits. Their protocol assumes the existence of a reliable
broadcast and a centralized trusted authority for key management distribution.

The Xor-trees approach of~\cite{Dolev:2000:XEA:354876.354877} extends
DC-Nets to achieve $O(n)$ amortized bit complexity, which is optimal.
In this protocol, only a single user is allowed to send at any one
time in a Xor-tree. Hence, the protocol is subject to performance
degradation due to collisions as the number of users increases. The
protocol assumes the existence of a public-key infrastructure and
a non-Byzantine polynomial-time adversary. The total bit complexity
of the protocol is $O(n^{2}t^{2})$ bits in worst case, where $t$
is the number of dishonest parties. The latency of the protocol is
$O(n)$ in worst case. However, a sender may broadcast large payloads
to amortize the costs. The amortized latency of the protocol is $O(1)$.

The Verdict protocol of~\cite{Verdict:2013:usenix} (which is based
on Dissent~\cite{dissent:10}) has a client-server architecture and
uses verifiable DC-Nets, where participants use public-key cryptography
to construct ciphertext, and knowledge proofs to detect and exclude
jamming parties before disruption. The protocol assumes the existence
of a few highly-available servers, where at least one server is honest.
All servers must be alive, however, for the protocol to work. An interesting
aspect of Verdict is that it is robust to a large fraction of Byzantine
parties (up to $n-2$). The paper demonstrates empirically that the
system scales well with the number of parties, when the number of
servers is fixed. The Tarzan protocol of Freedman and Morris~\cite{Freedman:2002:TPA:586110.586137}
provides resistance against traffic analysis, but only against a passive
adversary.

The Aqua protocol of Le Blond~et~al.~\cite{LeBlond:2013:TET:2486001.2486002}
provides $k$-anonymity with traffic-analysis resistance against passive
global attacks and active local attacks. The protocol achieves anonymity
in a way similar to Tor (onion routing) and achieves unobservability
through traffic obfuscating, which is to add artificial delay or artificial
traffic (called chaff) to the connection. 

In the last three decades, a large body of work has been devoted to
designing MPC protocols~\cite{bgw88,Beaver:1991,gennaro1998simplified,damgard2008scalable}.
Unfortunately, most of these protocols are inefficient and scale poorly
with the number of parties. Several MPC sorting and shuffling protocols
have been proposed in the literature~\cite{Goodrich:2011:RSS:2049697.2049701,Laur:2011:ObliviousDB,Zhang:2011,Hamada:2013:MPCSorting}.
Laur~et~al.~\cite{Laur:2011:ObliviousDB} describe a multi-party
shuffling protocol that can be used for anonymizing a set of inputs.
Although the communication and round complexity of their protocol
scales well with the number of inputs, they scale exponentially with
the number of parties and hence, the method cannot be used in our
model, where $n$ is relatively large. 

Goodrich\ \cite{Goodrich:2011:RSS:2049697.2049701} proposes an efficient
data-oblivious randomized shellsort algorithm. Unfortunately, when
implemented in a multi-party setting, this protocol requires $O(m)$
rounds of communication to sort $m$ values and has communication
complexity $O(\ell^{2}n^{2}m\log m)$, where $\ell$ is the message
size. Zhang~\cite{Zhang:2011} and Hamada~et~al.~\cite{Hamada:2013:MPCSorting}
develop constant-round MPC sorting protocols that scale well with
the number of inputs but scale poorly with the number of parties.

\section{Preliminaries \label{sec:preliminaries}}

In this section, we define standard terms, notation, and results used
throughout the paper.

\subsection{Notation}

An event occurs \emph{with high probability}, if it occurs with probability
at least $1-1/n^{c}$, for any $c>0$ and all sufficiently large $n$.
We denote the set of integers $\{1,...,n\}$ by $[n]$. Also, let
$\mathbb{Z}_{p}$ denote the additive group of integers modulo a prime
$p$.

\subsection{Basic Tools \label{sec:Basic-Tools}}

In this section, we review the definitions of standard basic tools
used throughout the paper.
\begin{description}
\item [{Verifiable~Secret~Sharing.}] An \emph{$(n,t)$-secret sharing
}scheme, is a protocol in which a dealer who holds a secret value
shares the secret among $n$ parties such that any set of $t$ parties
cannot gain any information about the secret, but any set of at least
$t+1$ parties can reconstructs it. An \emph{$(n,t)$-verifiable secret
sharing (VSS)} scheme is an \emph{$(n,t)$}-secret sharing scheme
with the additional property that after the sharing phase, a dishonest
dealer is either disqualified or the honest parties can reconstruct
the secret, even if shares sent by dishonest parties are spurious.
In our protocol, we use the constant-round VSS protocol of Katz~et~al.~\cite{Katz:2008:VSS}
that is based on Shamir's secret sharing scheme~\cite{shamir:how}.
This result is described in Theorem~\ref{thm:katz-vss}.%
\begin{comment}
For $n,t\in\mathbb{N}$, where $t<n$, an \emph{$(n,t)$-secret sharing
scheme} is a pair of algorithms \textsf{(Share, Reconstruct)} such
that a dealer runs \textsf{Share($s$)} to share a secret $s$ among
$n$ parties, and any subset of $t+1$ or more parties can compute
$s$ using \textsf{Reconstruct}, but no subset of $t$ or less parties
can. An \emph{$(n,t)$-non-interactive verifiable secret sharing scheme
(VSS)} is a group of algorithms \textsf{(Share, Reconstruct, Commit,
Verify)} such that \textsf{(Share, Reconstruct)} is an $(n,t)$-secret
sharing scheme, and (1) there is no polynomial-time strategy for picking
$t$ pieces of the secret, such that they can be used to predict the
secret, and (2) either \textsf{Reconstruct} outputs $s$ in which
case \textsf{Verify} outputs true, or honest parties conclude that
the dealer is malicious in which case \textsf{Verify} outputs false.

By \textsf{Commit}, the dealer commits to the set of shares that is
constructed via \textsf{Share,} and by \textsf{Verify}, the parties
verify the commitments.
\end{comment}
{} %
\begin{comment}
This protocol uses Shamir's secret sharing scheme~\cite{shamir:how}
along with a commitment scheme that is computationally-hiding under
the \emph{Discrete Logarithm (DL) assumption} and computationally-binding
under the \emph{$t$-Strong Diffie-Hellman ($t$-SDH) assumption}~\cite{Boneh:SDH:2004}
and $t$-polyDH assumption~\cite{asiacrypt-2010-23846}. Boneh and
Boyen~\cite{Boneh:SDH:2004} and Kate~et al.~\cite{asiacrypt-2010-23846}
show that solving $t$-SDH and $t$-polyDH problems with $t<O(\sqrt[3]{p})$
each requires $\Omega(\sqrt{p/t})$ expected time to be solved. 
\begin{thm}
\label{thm:evss-1-1}\emph{\cite{asiacrypt-2010-23846}} There exists
a synchronous $(n,t)$-non-interactive verifiable secret sharing scheme
for $n\geq2t+1$ secure under the DL, $t$-SDH, and $t$-polyDH assumptions.
In worst case, the protocol requires two broadcasts and four rounds
of communication.\end{thm}
\end{comment}
\end{description}
\begin{thm}
\label{thm:katz-vss}\emph{\cite{Katz:2008:VSS}} There exists a synchronous
linear\emph{ $(n,t)$}-VSS scheme for $t<n/3$ that is perfectly-secure
against a static active adversary. The protocol requires one broadcast
and three rounds of communication.\smallskip{}

\end{thm}
For practical purposes, one can use the cryptographic VSS scheme of
Kate~et al.~\cite{asiacrypt-2010-23846} called \emph{eVSS} (stands
for efficient VSS), which is based on Shamir's scheme and the hardness
of the Discrete Logarithm (DL) problem. Since eVSS generates commitments
over elliptic curve groups, it requires smaller message sizes than
other DL-based VSS scheme such as~\cite{gennaro1998simplified}.\smallskip{}

\begin{thm}
\label{thm:evss}\emph{\cite{asiacrypt-2010-23846}} There exists
a synchronous linear\emph{ $(n,t)$}-VSS scheme for $t<n/2$ that
is secure against a computationally-bounded static adversary. In worst
case, the protocol requires two broadcasts and four rounds of communication.\end{thm}
\begin{description}
\item [{Quorum~Building.}] A \emph{good quorum} is a set of $N=O(\log n)$
parties that contains a majority of honest parties. King et al.~\cite{ICDCN11}
showed that a \emph{Byzantine Agreement (BA) }protocol can be used
to bring all parties to agreement on a collection of $n$ good quorums.%
\begin{comment}
Similar to \cite{DKMS-ICDCN-2014}, we reduce the amount of communication
required for our protocol by creating $n$ groups of $N=O(\log n)$
parties called \emph{quorums}, where at most $T=(1/6-\epsilon)N$
of the parties in each quorum are dishonest. Scalability is achieved
by allowing parties of each quorum to communicate only with members
of the same quorum and members of a constant number of other quorums.
\end{comment}
{} In this paper, we use the fast BA protocol of Braud-Santoni~et~al.~\cite{Braud-Santoni:2013:FBA:2484239.2484243}
to build $n$ good quorums.

\begin{thm}
\label{thm:quorum-building}\emph{\cite{ICDCN11,Braud-Santoni:2013:FBA:2484239.2484243}}
There exists a constant-round unconditionally-secure protocol that
brings all good parties to agreement on $n$ good quorums with high
probability. The protocol has $\tilde{O}(n)$ amortized communication
and computation complexity%
\footnote{Amortized communication complexity is the total number of bits exchanged
divided by the number of parties.%
}, and it can tolerate up to $t<(1/3-\epsilon)n$ malicious parties.
\vspace{-5pt}

\end{thm}
\item [{Secure~Broadcast.}] In the malicious setting, when parties have
only access to secure pairwise channels, a protocol is required to
ensure secure (reliable) broadcast%
\footnote{We are not aware of any easy approach to physically implement a broadcast
channel without assuming a trusted party.%
}. Such a broadcast protocol guarantees all parties receive the same
message even if the broadcaster (dealer) is dishonest and sends different
messages to different parties. It is known that a BA protocol can
be used to perform secure broadcasts.%
\begin{comment}
\begin{description}
\item [{In}] our protocols, secure broadcasts are only performed among
small subsets of parties (quorums). Hence, we use the cryptographic
BA protocol of Cachin~et al.~\cite{Cachin00randomoracles} in quorums
as it is more efficient than the BA protocol of Braud-Santoni et~al.~\cite{Braud-Santoni:2013:FBA:2484239.2484243}
for small number of parties. The following theorem will be used in
our protocols.\end{description}
\begin{thm}
\emph{\cite{Cachin00randomoracles}} There exists a computationally-secure
protocol for performing reliable broadcasts in a network with $n$
parties connected via secure pairwise channels, where at most a 1/3
fraction of the parties are malicious. The protocol runs in constant
expected time and sends $\theta(n^{2})$ messages.\vspace{-5pt}

The BA algorithm of\ \cite{Cachin00randomoracles} requires a trusted
party to distribute cryptographic keys initially in order to setup
a public key infrastructure. In our protocol, parties jointly generate
a random seed in the setup phase, which is then used to generate a
common public key.\end{thm}
\end{comment}
{} In our protocol, we use the BA algorithm of Braud-Santoni~et al.~\cite{Braud-Santoni:2013:FBA:2484239.2484243}
to perform secure broadcasts. 

\begin{thm}
\emph{\label{thm:broadcast}\cite{Braud-Santoni:2013:FBA:2484239.2484243}}
There exists a constant-round unconditionally-secure protocol for
performing secure broadcasts among $n$ parties. The protocol has
$\tilde{O}(n)$ amortized communication and computation complexity,
and it can tolerate up to $t<(1/3-\epsilon)n$ malicious parties.\vspace{-3pt}

\end{thm}
\item [{Sorting~Networks.}] A \emph{sorting network} is a network of \emph{comparators}.
Each comparator is a gate with two input wires and two output wires.
When two values enter a comparator, it outputs the lower value on
the top output wire, and the higher value on the bottom output wire.
Ajtai~et al.~\cite{Ajtai:1983:SCL:61981.61982} describe an asymptotically-optimal
$O(\log n)$ depth sorting network. However, this network is not practical
due to large constants hidden in the depth complexity. Leighton and
Plaxton~\cite{Leighton:1990:Sorting} propose a practical \emph{probabilistic
sorting circuit} that sorts with very high probability meaning that
it sorts all but $\epsilon\cdot n!$ of the $n!$ possible input permutations,
where $\epsilon=1/2^{2^{k\sqrt{\log n}}}$, for any constant $k>0$.
For example, for $k=2$ and $n>64$, we get $\epsilon<10^{-9}$. While
this circuit is sufficient for us to prove our main results (Theorem~\ref{thm:main}),
one can instead use the $O(\log^{2}n)$-depth sorting network of Batcher~\cite{Batcher:1968:SNA:1468075.1468121}
for sorting all $n!$ permutations at the expense of $O(\log^{2}n)$
protocol latency.
\item [{Secure~Comparison.}] Given two linearly secret-shared values $a,b\in\zp$,
Nishide and Ohta~\cite{Nishide:PKC:2007} propose an efficient protocol
for computing a sharing of $\rho\in\{0,1\}$ such that $\rho=(a\leq b)$.
Their protocol has $O(1)$ rounds and requires $O(\ell)$ invocations
of a secure multiplication protocol, where $\ell$ is the bit-length
of elements to be compared. We refer to this protocol by \textsf{Compare}.
We also describe a fast multiplication protocol to be used along with
the comparison protocol of \cite{Nishide:PKC:2007} for implementing
fast comparator gates.
\item [{Share~Renewal.}] In our protocol, a shuffling circuit is securely
evaluated. Each gate of the circuit is assigned a quorum $Q$ and
the parties in $Q$ are responsible for comparison of secret-shared
inputs. Then, they send the secret-shared result to any quorums associated
with gates that need this result as input. Let $Q^{\prime}$ be one
such quorum. In order to secret-share the result to $Q^{\prime}$
without revealing any information to any individual party (or to any
coalition of dishonest parties), a fresh sharing of the result must
be distributed in $Q^{\prime}$. To this end, we use the share renewal
protocol of Herzberg~et al.~\cite{Herzberg:1995:proactive}%
\footnote{This technique was first used by Ben-Or~et~al.~\cite{bgw88} and
was later proved to be UC-secure by Asharov and Lindell in \cite{cryptoeprint:2011:136}.%
}. To update a shared value $s$ defined over a degree $d$ polynomial
$\phi(x)$, the protocol generates a degree $d$ random polynomial
$\delta(x)$ such that $\delta(0)=0$. The new polynomial $\phi^{*}(x)$
is then computed from $\phi^{*}(x)=\phi(x)+\delta(x)$. Since $\phi^{*}(0)=\phi(0)+\delta(0)=s$,
the new polynomial $\phi^{*}(x)$ defines a fresh sharing of $s$.
Combined with the VSS scheme of Theorem\ \ref{thm:katz-vss} (or
Theorem\ \ref{thm:evss}), this protocol is secure against an active
adversary (with $t<n/3$), takes constant number of rounds, and requires
each party to send $O(n^{2})$ field elements. We refer to this protocol
by \textsf{RenewShares }throughout this paper.
\end{description}

\section{Our Protocol \label{sec:protocol}}

In this section, we describe our protocol for anonymous broadcast.
Consider $n$ parties $P_{1},P_{2},...,P_{n}$ each having a message
$x_{i}\in\mathbb{Z}_{p}$, for a prime $p$ and all $i\in[n]$\textit{\emph{.
}}The parties want to anonymously send the messages to each other
and receive the results back. We first describe an ideal functionality
of our protocol where a hypothetical trusted party $P$ computes the
desired protocol outcome by communicating with all parties%
\footnote{We are inspired by the standard ideal/real world definition for multi-party
protocols proposed by Canetti~\cite{Canetti:UCSecurity:2001}. %
}. In every run of the protocol, $P$ executes a shuffling protocol
over the messages and sends the shuffled sequence of messages to all
parties. The shuffling protocol first chooses a uniform random number
$r_{i}\in\mathbb{Z}_{p}$ to form an input pair $(r_{i},x_{i})$ for
each party $P_{i}$ and for all $i\in[n]$. The protocol then uses
a shuffling circuit that is based on the sorting network of Leighton
and Plaxton~\cite{Leighton:1990:Sorting} to sort the set of pairs
$\{(r_{1},x_{1}),...,(r_{n},x_{n})\}$ according to their first elements.
We later show that this functionality randomly permutes the set of
inputs $\{x_{1},...,x_{n}\}$.

In our protocol, we denote the shuffling circuit by $C$, which has
$m$ gates. Each gate is essentially a comparator gate and is denoted
by $G_{i}$, for $i\in[m].$ Our protocol first creates $n$ quorums
and then assigns each gate of $C$ to a quorum. For $n$ parties,
the circuit has $\lceil n/2\rceil$ input gates as each comparator
gate has two inputs. We label the quorums assigned to the input gates
by $Q_{1},...,Q_{\lceil n/2\rceil}$ and call them \emph{input quorums}.
$C$ also has $\lceil n/2\rceil$ output gates, which correspond to
\emph{output quorums} labeled by $Q_{1}^{\prime},...,Q_{\lceil n/2\rceil}^{\prime}$.

We now implement the real functionality of our protocol based on the
ideal functionality described above. Protocol~\ref{pro:Main} defines
our main algorithm. Throughout the protocol, we represent each\emph{
}shared value $s\in\zp$ by $\la s\ra=(s_{1},...,s_{n})$ meaning
that each party $P_{i}$ holds a share $s_{i}$ generated by the VSS
scheme during its sharing phase. Using the natural component-wise
addition of representations, we define $\la a\ra+\la b\ra=\la a+b\ra$.
For multiplication, we define $\la a\ra\cdot\la b\ra=\mathsf{Multiply(}\la a\ra,\la b\ra\mathsf{)}$,
where $\mathsf{Multiply}$ is a protocol defined in this section.

For the reader's convenience, we now list all subprotocols that we
use in our protocol. These subprotocols are also defined in this section,
Section~\ref{sec:preliminaries}, and Section~\ref{sec:Remaining-Algorithms}.
\begin{itemize}
\item \textsf{GenRand: }A well-known technique (due to \cite{Beaver:1991})
that generates uniform random secrets by adding shared values chosen
uniformly at random by each party. Defined as Protocol\ \ref{pro:GenRand}
in Section~\ref{sec:Remaining-Algorithms}.
\item \textsf{Compare: }The protocol of Nishide and Ohta\ \cite{Nishide:PKC:2007}
for securely comparing two secret-shared values.
\item \textsf{Multiply: }Multiplies two secret-shared values. Defined as
Protocol~\ref{pro:Multiply} in this section.
\item \textsf{RenewShares: }The share renewal protocol of Herzberg~et al.~\cite{Herzberg:1995:proactive}
for re-randomizing shares of a secret-shared value. This makes past
knowledge of the adversary obsolete after re-sharing the secret.
\item \textsf{Reconst: }Given a shared secret $s$, reconstructs $s$ via
polynomial interpolation and error correction. Defined as Protocol\ \ref{pro:Reconst}
in Section~\ref{sec:Remaining-Algorithms}.
\end{itemize}
{\small{}}
\begin{algorithm}[H]
{\small{\caption{\enskip{}\textsf{AnonymousBroadcast}}
\label{pro:Main}}}{\small \par}

{\small{\medskip{}
}}{\small \par}
\begin{enumerate}
\item \textbf{\small{Setup:}}{\small{ Parties jointly run the quorum building
protocol of Theorem~\ref{thm:quorum-building} with all parties to
agree on $n$ good quorums $Q_{1},...,Q_{n}$, and assign each gate
$G_{i}$ to the $(i\bmod n)$-th quorum, for all $i\in[m]$.}}%
\begin{comment}
{\small{Each party in each quorum $Q$ jointly runs the key generation
protocol of eVSS with other parties in $Q$.}}
\end{comment}
{\small \par}
\item \textbf{\small{Input~Sharing:}}{\small{ Parties $P_{i}$ and $P_{i+1}$
secret share their inputs $x_{i}$ and $x_{i+1}$ among all parties
of $Q_{\lceil i/2\rceil}$ using the VSS scheme of Theorem~\ref{thm:katz-vss},
for all $i\in[n]$. }}{\small \par}
\item \textbf{\small{Random~Generation: }}{\small{Parties in each input
quorum $Q_{\lceil i/2\rceil}$ run $\mathsf{GenRand}$ twice to generate
sharings $\la r_{i}\ra$ and $\la r_{i+1}\ra$ of two random elements
$r_{i},r_{i+1}\in\zq$, for some prime $q$. At the end of this step,
$Q_{\lceil i/2\rceil}$ holds two pairs of sharings $(\la r_{i}\ra,\la x_{i}\ra)$
and $(\la r_{i+1}\ra,\la x_{i+1}\ra)$, for all $i\in[n]$.}}{\small \par}
\item \textbf{\small{Circuit~Computation: }}{\small{The circuit is evaluated
level-by-level starting from the input gates. For each gate $G$ and
the quorum $Q$ associated with it, parties in $Q$ do the following.\medskip{}
}}{\small \par}

\begin{enumerate}
\item \textbf{\small{Gate~Computation: }}{\small{Let $(\la r\ra,\la x\ra)$
and $(\la r^{\prime}\ra,\la x^{\prime}\ra)$ be the sharings associated
with the inputs of $G$. Parties compute}}\textsf{\small{ $\la\rho\ra=\mathsf{Compare(}\la r\ra,\la r^{\prime}\ra\mathsf{)}$}}{\small{,
where $\rho=(r\leq r^{\prime})$. Then, they compute the output pairs
$(\la s\ra,\la y\ra)$ and $(\la s^{\prime}\ra,\la y^{\prime}\ra)$
from\vspace{-8pt}
\begin{gather}
\begin{aligned}\la s\ra & =\la\rho\ra\cdot\la r\ra+(1-\la\rho\ra)\cdot\la r^{\prime}\ra\\
\la y\ra & =\la\rho\ra\cdot\la x\ra+(1-\la\rho\ra)\cdot\la x^{\prime}\ra
\end{aligned}
\qquad\begin{aligned}\la s^{\prime}\ra & =\la\rho\ra\cdot\la r^{\prime}\ra+(1-\la\rho\ra)\cdot\la r\ra\\
\la y^{\prime}\ra & =\la\rho\ra\cdot\la x^{\prime}\ra+(1-\la\rho\ra)\cdot\la x\ra
\end{aligned}
\label{eq:switching}
\end{gather}
\vspace{-10pt}
}}{\small \par}
\item []\setcounter{enumii}{1}{\small{For every addition over two shared
values $\la a\ra$ and $\la b\ra$ performed above, parties computes
$\la c\ra=\la a\ra+\la b\ra$. For every multiplication, they run
$\la c\ra=\mathsf{Multiply(}\la a\ra,\la b\ra\mathsf{)}$.\smallskip{}
}}{\small \par}
\item \textbf{\small{Output~Resharing:}}{\small{ Parties run $\mathsf{RenewShares}$
}}over $\la s\ra$, $\la y\ra$, $\la s^{\prime}\ra$, and $\la y^{\prime}\ra$
{\small{to re-share them in}} {\small{the quorum associated with the
parent gate. }}%
\begin{comment}
{\small{Parties in $Q$ run }}\textsf{\small{$\mathsf{Reshare(}\la s\ra,\la y\ra,Q^{\,\prime}\mathsf{)}$}}{\small{
and }}\textsf{\small{$\mathsf{Reshare(}\la s^{\,\prime}\ra,\la y^{\,\prime}\ra,Q^{\,\prime}\mathsf{)}$}}{\small{,
where $Q^{\,\prime}$ is the quorum associated with the parent gate.}}
\end{comment}
{\small \par}
\end{enumerate}
\item \textbf{\small{Output~Propagation: }}{\small{Let $(\la s_{i}\ra,\la y_{i}\ra)$
and $(\la s_{i+1}\ra,\la y_{i+1}\ra)$ be the pairs of shared values
each output quorum $Q_{i}^{\prime}$ receives. Parties in $Q_{i}^{\prime}$
run }}\textsf{\small{$z_{i}\mathsf{=Reconst(}\la y_{i}\ra\mathsf{)}$}}{\small{
and }}\textsf{\small{$z_{i+1}\mathsf{=Reconst(}\la y_{i+1}\ra\mathsf{)}$
}}{\small{and send $(z_{i}$,$z_{i+1})$ to all $n$ parties. Each
party receiving a set of $N$ pairs from each output quorum, chooses
one pair via majority filtering and considers it as the output of
that quorum.}}\end{enumerate}
\end{algorithm}
{\small \par}

We prove the correctness and secrecy of Protocol~\ref{pro:Main}\textsf{
}(and Theorem~\ref{thm:main}) in Section~\ref{sub:Security-of-Main}.
In Lemma~\ref{lem:Collision}, we show that for sufficiently large
$k>0$ and $p>\frac{3}{2}kn^{2}\log n$, this protocol\textsf{ }computes
a random permutation of the input messages with high probability.
We also prove the following lemmas in Section~\ref{sec:Proof-of-Costs}.\medskip{}

\begin{lem}
\emph{\label{lem:protocol-costs}Using the perfectly-secure VSS of
Theorem~\ref{thm:katz-vss}, Protocol~\ref{pro:Main}} sends $\tilde{O}(n)$
bits, computes $\tilde{O}(n)$ operations, and takes $O(\log n)$
rounds of communication for shuffling $n$ messages.
\end{lem}
{\small{\medskip{}
}}{\small \par}
\begin{lem}
\emph{\label{lem:protocol-costs-eVSS}Using the cryptographic VSS
of Theorem~\ref{thm:evss}, Protocol~\ref{pro:Main}} sends $O(n\log^{5}n)$
messages of size $O(\kappa+\log n)$ and computes $O(n(\kappa+\log n)\log^{5}n)$
operations for shuffling $n$ messages, where $\kappa$ is the security
parameter. The protocol has $O(\log n)$ rounds of communication.\medskip{}

\end{lem}
We now describe the subprotocol $\mathsf{Multiply}$ that is based
on a well-known technique proposed by Beaver~\cite{Beaver:1991}.
The technique generates a shared multiplication triple $(\la u\ra,\la v\ra,\la w\ra)$
such that $w=u\cdot v$. The triple is then used to convert multiplications
over shared values into additions. Since the triple generation step
is independent of the inputs, it can be done in a preprocessing phase.
Moreover, since each triple acts like a one-time pad, a fresh triple
must be used for each multiplication.

{\small{}}
\begin{algorithm}[H]
{\small{\caption{\enskip{}\textsf{Multiply}}
\label{pro:Multiply}}}{\small \par}

\textit{\small{Usage.}}{\small{ Initially, parties hold two secret-shared
values $\la a\ra=(a_{1},...,a_{N})$ and $\la b\ra=(b_{1},...,b_{N})$
that are on a polynomial of degree $N/3$. The protocol computes a
shared value $\la c\ra=(c_{1},...,c_{N})$ such that $c=a\cdot b$.}}{\small \par}

{\small{\medskip{}
}}{\small \par}

{\small{$\mathsf{Multiply(}\la a\ra,\la b\ra\mathsf{):}$\medskip{}
}}{\small \par}

{\small{For all $i\in[N]$, party $P_{i}$ performs the following
steps.}}{\small \par}
\begin{enumerate}
\item {\small{$P_{i}$ runs }}\textsf{\small{GenRand}}{\small{ twice with
other parties to generate two shared random values $u_{i}$ and $v_{i}$
both on polynomials of degree $N/6$. Then, he computes $w_{i}=u_{i}\cdot v_{i}$ }}{\small \par}
\item {\small{$P_{i}$ updates $u_{i}$, $v_{i}$, and $w_{i}$ by running
$\mathsf{RenewShares}$. }}{\small \par}
\item {\small{$P_{i}$ computes $\varepsilon_{i}=a_{i}+u_{i}$ and $\delta_{i}=b_{i}+v_{i}$
and runs $\mathsf{\varepsilon=Reconst(}\varepsilon_{i}\mathsf{)}$
and $\delta=\mathsf{Reconst(}\delta_{i}\mathsf{)}$. Then, he computes
$c_{i}=w_{i}+\delta a_{i}+\varepsilon b_{i}-\varepsilon\delta$.}}\end{enumerate}
\end{algorithm}
{\small \par}

Clearly, $\varepsilon$ and $\delta$ can be safely revealed to all
parties so that each party can compute $\varepsilon\delta$ locally.
The correctness and the secrecy of the algorithm\textsf{ }are proved
by Beaver~\cite{Beaver:1991}. The only difference between Beaver's
original method and our protocol is as follows. Following \cite{bgw88},
he proposes to generate shared random elements $u$ and $v$ on a
polynomial of degree $N/3$ and multiply them to get a polynomial
of degree $2N/3$ for $w$. Then, a degree reduction algorithm is
run to reduce the degree from $2N/3$ to $N/3$. Instead of this,
we choose polynomials of degree $N/6$ for $u$ and $v$ to get a
polynomial of degree $N/3$ for $w$. In our protocol, since up to
$N/6$ of parties are dishonest is each quorum, we can do this without
revealing any information to the adversary. We prove this formally
in Lemma \ref{Lem:eVSSProps} and Lemma \ref{lem:Multiply-Secrecy}. 

One issue with multiplying shares, as noticed by Ben-Or~et~al.~\cite{bgw88},
is that the resulting polynomial is not completely random because,
for example, the product of two polynomials cannot be irreducible.
To solve this, we simply re-randomize the shares via \textsf{RenewShare}.
Moreover, since $\la u\ra$ and $\la v\ra$ are on degree $N/6$ polynomials
and \textsf{Multiply }uses them to mask-and-reveal the inputs (that
are on degree $N/3$ polynomials) in the last step, it is necessary
to randomize $\la u\ra$ and $\la v\ra$ via \textsf{RenewShares}
before masking. This puts the shared values $\la u\ra$ and $\la v\ra$
on new random polynomials with degree $N/3$.%
\begin{comment}
We now describe the subprotocol \textsf{Output}, which distributes
the final results (i.e., circuit outputs) among all $n$ parties.
The naive approach is to let all parties in each output quorum send
their values to all $n$ parties. Unfortunately, this approach is
not load-balanced since it requires each party in an output quorum
to send $O(n)$ messages. Instead, we propose a simple algorithm where
each party sends only $O(\log^{2}n)$ messages we build a binary tree
of depth $\log n$, where each of the $n$ quorums is assigned to
exactly one of the nodes. Then, 
\end{comment}
\medskip{}

\begin{defn}
\textbf{{[}Perfect Random Permutation{]}} Consider a set of $n\geq1$
elements $A=\{a_{1},a_{2},...a_{n}\}$. A perfect random permutation
of $A$ is a permutation chosen uniformly at random from the set of
all possible $n!$ permutations of $A$.\medskip{}
\end{defn}
\begin{lem}
\label{lem:permutation} Consider a sequence of input pairs $(r_{1},x_{1}),...,(r_{n},x_{n})$,
and a sorting protocol $\Pi$\textsf{ }that sorts the pairs according
to their first elements. $\Pi$ computes a perfect random permutation
of the pairs if their first elements are chosen uniformly at random
and are distinct.\end{lem}
\begin{proof}
Let $X=(r_{1},x_{1}),...,(r_{n},x_{n})$ be the input sequence and
$Y=(s_{1},y_{1}),...,(s_{n},y_{n})$ be the output sequence of $\Pi$.
Note $s_{1},...,s_{n}$is the sorted sequence of $\{r_{1},...,r_{n}\}$.
An arbitrary output sequence of pairs $Y^{\prime}=(s_{1}^{\prime},y_{1}^{\prime}),...,(s_{n}^{\prime},y_{n}^{\prime})$
is said to be equal to $Y$ if $y_{i}=y{}_{i}^{\prime}$, for all
$i\in[n]$. We want to prove that the probability of $Y^{\prime}$
being equal to $Y$ is $\frac{1}{n!}$. In general, for any $i\in[n]$,
$y_{i}=y{}_{i}^{\prime}$ if and only if $s{}_{i}^{\prime}$ is the
$i$-th smallest element in $\{r_{1},...,r_{n}\}$ conditioned on
knowing the $i-1$ smallest elements, which happens with probability
$\frac{1}{n-i+1}$. Thus, the probability that $Y=Y^{\prime}$ is
$\frac{1}{n}\cdot\frac{1}{n-1}\cdot...\cdot\frac{1}{2}\cdot1=\frac{1}{n!}.$
\end{proof}
In the random generation step of Protocol~\ref{pro:Main}, it is
possible that two or more input quorums choose the same random elements
from $\zq$. In this situation, we say a \emph{collision} happens.
Collisions reduce the level of anonymity our protocol guarantees because
the higher the probability of collisions, the higher the chance the
adversary is given in guessing the correct sequence of inputs. On
the other hand, we observe that if the field size (i.e., $p$) is
sufficiently large, then the probability of collisions becomes overwhelmingly
small. Lemma~\ref{lem:Collision} gives a lower bound on $p$ such
that collisions are guaranteed to happen with negligible probability.\medskip{}

\begin{lem}
\label{lem:Collision} Let $\zq$ be the field of random elements
generated in the random generation step of Protocol~\ref{pro:Main}.
The probability there is a collision between any two parties is negligible
if $p>\frac{3}{2}kn^{2}\log n$, for some $k>0$. \end{lem}
\begin{proof}
Based on Theorem~\ref{thm:quorum-building}, all input quorums are
good. Based on the correctness of \textsf{GenRand}, all elements generated
by the input quorums in the random generation step of Protocol~\ref{pro:Main}\textsf{
}are chosen uniformly at random and independent of all other random
elements generated throughout the protocol. Let $P_{i}$ and $P_{j}$
be two parties and $r_{i}$ and $r_{j}$ be the random values assigned
to them respectively by their corresponding input quorums. The probability
that $r_{i}=r_{j}$ is $1/p$. Let $X_{ij}$ be the following indicator
random variable and $Y$ be a random variable giving the number of
collisions between any two parties,{\small{
\[
X_{ij}=\begin{cases}
1, & r_{i}=r_{j}\\
0, & otherwise
\end{cases},\quad Y=\sum_{i,j\in[n]}X_{ij}.
\]
}}Using the linearity of expectations,{\small{
\[
E(Y)=E\big(\sum_{i,j\in[n]}X_{ij}\big)=\sum_{i,j\in[n]}E(X_{ij})=\frac{1}{p}\left(_{2}^{n}\right)=\frac{n(n-1)}{2p}.
\]
}}We want to find an upper bound on the probability of collisions
using the Chernoff bound defined by{\small{
\[
Pr(Y\geq(1+\alpha)E(Y))\leq e^{-\frac{\alpha^{2}E(Y)}{3}}.
\]
}}To ensure that no collision happens with high probability, we need
to have $(1+\alpha)E(Y)<1$ while $e^{-\frac{\alpha^{2}E(Y)}{3}}<\frac{1}{n^{k}}$,
for any $k>0$. Choosing $\alpha<\frac{1}{E(Y)}-1$ and solving the
inequalities for $E(Y)$ we get\medskip{}

\noindent $e^{-\frac{\alpha^{2}E(Y)}{3}}<\frac{1}{n^{k}}\quad\Rightarrow\quad e^{-\frac{\alpha^{2}E(Y)}{3}}<e^{-k\log n}\quad\Rightarrow\quad1-\frac{\alpha^{2}E(Y)}{3}<-k\log n\quad\Rightarrow\quad\frac{(\alpha+1)^{2}E(Y)}{3}>k\log n\quad\Rightarrow$

\smallskip{}

\noindent $\frac{1}{3E(Y)}>k\log n\quad\Rightarrow\quad E(Y)<\frac{1}{3k\log n}$

\noindent \vspace{0pt}

\noindent Since $E(Y)=\frac{n(n-1)}{2p}<\frac{1}{3k\log n}$ , solving
this for $p$ gives the bound $p>\frac{3}{2}kn^{2}\log n$ and $\alpha<3k\log n-1$.\vspace{0cm}

\end{proof}

\section{Simulation Results \label{sec:simulations}}

To study the feasibility of our scheme and compare it to previous
work, we implemented an experimental simulation of our protocol and
three other protocols which can be used for shuffling $n$ inputs
randomly (with traffic-analysis resistance) in the same setting. These
protocols are due to Dani~et~al.~\cite{DKMS-ICDCN-2014}, Boyle~et~al.~\cite{Boyle:2013:CLS:2450206.2450227},
and Zamani~et~al.~\cite{Zamani:2013:FOCI}. To the best of our
knowledge, these protocols have the best scalability with respect
to the network size among other works described in the literature
for the same setting. Since the protocols of \cite{DKMS-ICDCN-2014}
and \cite{Boyle:2013:CLS:2450206.2450227} are general MPC algorithms,
we use them for running our shuffling technique described in Section~\ref{sec:protocol}.
The protocol of \cite{Zamani:2013:FOCI} is a jamming-resistant version
of DC-Nets that scales better than other DC-Net protocols~\cite{chaum88,Juels04,vonAhn03,Waidner:1990:DCD:111563.111630}.
We stress that we are interested in evaluating our protocols for large
network sizes and hence, our choice of protocols for this section
is based on their scalability for large $n$'s. One may find other
protocols in the literature that perform better than the protocols
of our choice for small $n$'s.

We run our protocol for inputs chosen from $\mathbb{Z}_{p}$ with
a 160-bit prime $p$ for getting about 80 bits of security. We set
the parameters of our protocol in such a way that we ensure the probability
of error for the quorum building algorithm of \cite{Braud-Santoni:2013:FBA:2484239.2484243}
is smaller than $10^{-5}$. For the sorting circuit, we set $k=2$
to get $\epsilon<10^{-8}$ for all values of $n$ in the experiment.
Clearly, for larger values of $n$, the error becomes superpolynomially
smaller, e.g., for $n=2^{25}$, we get $\epsilon<10^{-300}$. For
all protocols evaluated in this section, we assume cheating (by malicious
parties) happens in every round of the protocols. This is essential
for evaluating various strategies used by these protocols for tolerating
active attacks. 

Figure~\ref{fig:plots} illustrates the simulation results obtained
for various network sizes between $2^{5}$ and $2^{30}$ (i.e., between
32 and about 1 billion). To better compare the protocols, the vertical
and horizontal axis of the plot are scaled logarithmically. The $x$-axis
presents the number of parties and the $y$-axis presents the number
of Kilobytes sent by each party for delivering one anonymous bit.
In this figure, we report results from three different versions of
our protocols. The first plot (marked with circles) belongs to our
unconditionally-secure protocol that uses the perfectly-secure VSS
scheme of Katz~et~al.~\cite{Katz:2008:VSS}. The second plot (marked
with stars) represents our computationally-secure protocol which uses
the cryptographic VSS of Kate~et~al.~\cite{asiacrypt-2010-23846}.
The last plot (marked with diamonds) shows the cost of the cryptographic
protocol with amortized (averaged) setup cost. To obtain the amortized
plot, we run the setup phase of Protocol~\ref{pro:Main} once and
then used the setup data to run our online protocol 100 times. The
total number of bits sent was then divided by 100 to get the average
communication cost. To achieve better results, we also generated a
sufficient number of random triples in the setup phase. Then, the
triples were used by our multiplication subprotocol in the online
phase to multiply secret-shared values efficiently.

We observe that our protocols (even the unconditional version) perform
significantly better than the other protocols. For example, for $n=2^{20}$
(about 1 million parties%
\footnote{This is less than 1\% of the number of active Twitter users. An intriguing
application of our protocol is an anonymous version of Twitter. %
}), the amortized protocol requires each party to send about 64KB of
data per anonymous bit delivered (about 8MB for our crypto version
and about 64MB for our unconditional version) while the protocols
of \cite{Boyle:2013:CLS:2450206.2450227}, \cite{DKMS-ICDCN-2014},
and \cite{Zamani:2013:FOCI} each send more than one Terabytes of
data per party and per anonymous bit delivered.

\begin{figure}[t]
\begin{centering}
~~~~~~~~~\includegraphics[scale=0.5]{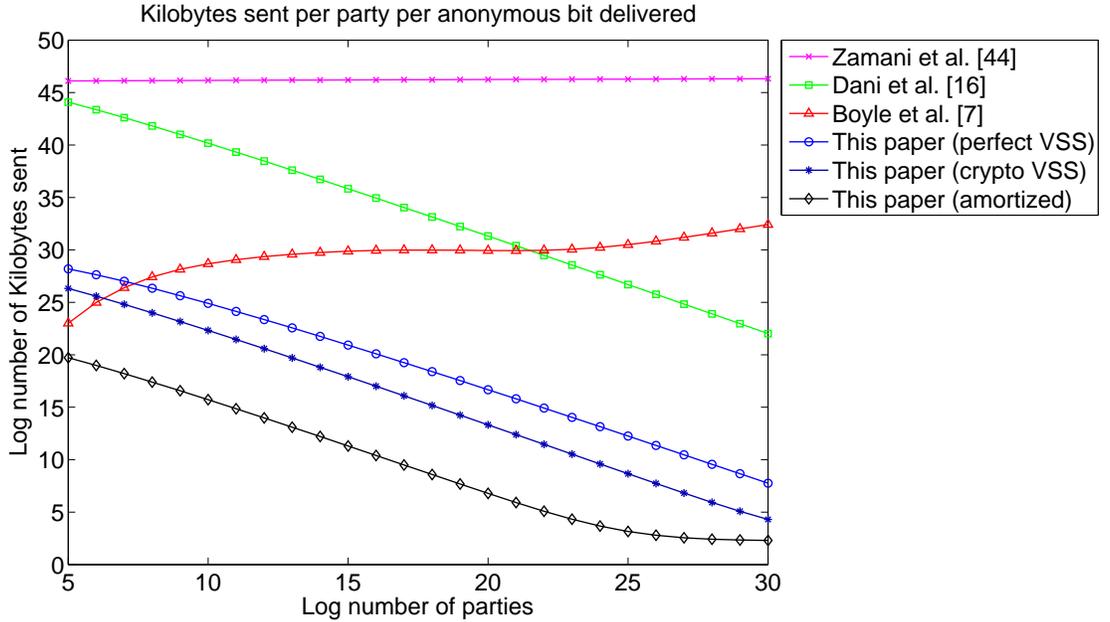}
\par\end{centering}

\caption{Communication cost of anonymous broadcast protocols.}

\label{fig:plots} 
\end{figure}

\section{Conclusion and Open Problems \label{sec:conclusion}}

We described an anonymous broadcast protocol that is fully decentralized
and tolerates up to $n/6$ active faults. Moreover, our protocol is
load-balanced and can tolerate traffic-analysis attacks. The amount
of information sent and the amount of computations performed by each
party scales polylogarithmically with the number of parties. The scalability
is achieved by performing local communications and computations in
groups of logarithmic size and by relaxing the latency requirements. 

Several open problems remain. First, can we decrease the number of
rounds of our protocol using a smaller-depth sorting circuit? For
example, since our protocol sorts uniform random numbers, it seems
possible to use a smaller depth non-comparison-based sorting circuit
like bucket sort. Second, can we improve performance even further
by detecting and blacklisting parties that exhibit adversarial behavior?
Finally, can we adopt our results to the asynchronous model of communication?
We believe that this is possible for a suitably chosen upper bound
on the fraction of faulty parties.

{\small{\bibliographystyle{abbrv}
\bibliography{C:/Users/Mahdi/Desktop/Dropbox/Svn_Research/Writeups/Common/security}

\begin{thebibliography}{10}

\bibitem{Ajtai:1983:SCL:61981.61982}
M.~Ajtai, J.~Koml\'{o}s, and E.~Szemer{\'e}di.
\newblock Sorting in $c\log{n}$ parallel steps.
\newblock {\em Combinatorica}, 3(1):1--19, Jan. 1983.

\bibitem{cryptoeprint:2011:136}
G.~Asharov and Y.~Lindell.
\newblock A full proof of the {BGW} protocol for perfectly-secure multiparty
  computation.
\newblock Cryptology ePrint Archive, Report 2011/136, 2011.

\bibitem{Batcher:1968:SNA:1468075.1468121}
K.~E. Batcher.
\newblock Sorting networks and their applications.
\newblock In {\em Proceedings of the April 30--May 2, 1968, spring joint
  computer conference}, AFIPS '68 (Spring), pages 307--314, New York, NY, USA,
  1968. ACM.

\bibitem{Beaver:1991}
D.~Beaver.
\newblock Efficient multiparty protocols using circuit randomization.
\newblock In J.~Feigenbaum, editor, {\em Advances in Cryptology -- CRYPTO '91},
  volume 576 of {\em Lecture Notes in Computer Science}, pages 420--432.
  Springer Berlin Heidelberg, 1991.

\bibitem{bgw88}
M.~{Ben-Or}, S.~Goldwasser, and A.~Wigderson.
\newblock Completeness theorems for non-cryptographic fault-tolerant
  distributed computing.
\newblock In {\em Proceedings of the Twentieth ACM Symposium on the Theory of
  Computing (STOC)}, pages 1--10, 1988.

\bibitem{Berlekamp:Welch:1986}
E.~Berlekamp and L.~Welch.
\newblock Error correction for algebraic block codes, {US Patent} 4,633,470,
  Dec. 1986.

\bibitem{Boyle:2013:CLS:2450206.2450227}
E.~Boyle, S.~Goldwasser, and S.~Tessaro.
\newblock Communication locality in secure multi-party computation: how to run
  sublinear algorithms in a distributed setting.
\newblock In {\em Proceedings of the 10\textsuperscript{th} theory of
  cryptography conference on Theory of Cryptography}, TCC'13, pages 356--376,
  Berlin, Heidelberg, 2013. Springer-Verlag.

\bibitem{Braud-Santoni:2013:FBA:2484239.2484243}
N.~{Braud-Santoni}, R.~Guerraoui, and F.~Huc.
\newblock Fast {B}yzantine agreement.
\newblock In {\em Proceedings of the 2013 ACM Symposium on Principles of
  Distributed Computing}, PODC~'13, pages 57--64, New York, NY, USA, 2013. ACM.

\bibitem{Canetti:UCSecurity:2001}
R.~Canetti.
\newblock Universally composable security: a new paradigm for cryptographic
  protocols.
\newblock In {\em Proceedings of the 42nd Annual Symposium on Foundations of
  Computer Science}, FOCS '01, pages 136--145, Oct 2001.

\bibitem{chaum:81}
D.~Chaum.
\newblock Untraceable electronic mail, return addresses, and digital
  pseudonyms.
\newblock {\em Commun. ACM}, 24(2):84--90, Feb. 1981.

\bibitem{chaum88}
D.~Chaum.
\newblock The dining cryptographers problem: Unconditional sender and recipient
  untraceability.
\newblock {\em Journal of Cryptology}, 1:65--75, 1988.

\bibitem{Chor:1998:PIR:293347.293350}
B.~Chor, E.~Kushilevitz, O.~Goldreich, and M.~Sudan.
\newblock Private information retrieval.
\newblock {\em J. ACM}, 45(6):965--981, Nov. 1998.

\bibitem{dissent:10}
H.~Corrigan-Gibbs and B.~Ford.
\newblock Dissent: accountable anonymous group messaging.
\newblock In {\em Proceedings of the 17th ACM Conference on Computer and
  Communications Security}, CCS~'10, pages 340--350, New York, NY, USA, 2010.
  ACM.

\bibitem{Verdict:2013:usenix}
H.~Corrigan-Gibbs, D.~I. Wolinsky, and B.~Ford.
\newblock Proactively accountable anonymous messaging in verdict.
\newblock In {\em Proceedings of the 22\textsuperscript{nd} USENIX Security
  Symposium}, pages 147--162, Berkeley, CA, USA, 2013.

\bibitem{damgard2008scalable}
I.~Damg{\aa}rd, Y.~Ishai, M.~Kr{\o}igaard, J.~Nielsen, and A.~Smith.
\newblock Scalable multiparty computation with nearly optimal work and
  resilience.
\newblock {\em Advances in Cryptology -- CRYPTO '08}, pages 241--261, 2008.

\bibitem{DKMS-ICDCN-2014}
V.~Dani, V.~King, M.~Movahedi, and J.~Saia.
\newblock Quorums quicken queries: Efficient asynchronous secure multiparty
  computation.
\newblock In {\em Distributed Computing and Networking}, volume 8314 of {\em
  Lecture Notes in Computer Science}, pages 242--256. Springer Berlin
  Heidelberg, 2014.

\bibitem{dingledine:2004}
R.~Dingledine, N.~Mathewson, and P.~Syverson.
\newblock Tor: the second-generation onion router.
\newblock In {\em Proceedings of the 13th USENIX Security Symposium}, pages
  21--21, Berkeley, CA, USA, 2004. USENIX Association.

\bibitem{Dolev:2000:XEA:354876.354877}
S.~Dolev and R.~Ostrovsky.
\newblock Xor-trees for efficient anonymous multicast and reception.
\newblock {\em ACM Trans. Inf. Syst. Secur.}, 3(2):63--84, May 2000.

\bibitem{Feigenbaum:2014:Panopticon}
J.~Feigenbaum and B.~Ford.
\newblock Seeking anonymity in an {Internet} panopticon.
\newblock {\em e-Print arXiv:1312.5307}, March 2014.

\bibitem{stajano:99}
S.~Frank and R.~Anderson.
\newblock The cocaine auction protocol: On the power of anonymous broadcast.
\newblock In {\em Proceedings of the Third International Workshop on
  Information Hiding}, IH 99, pages 434--447, London, UK, 2000.
  Springer-Verlag.

\bibitem{Freedman:2002:TPA:586110.586137}
M.~J. Freedman and R.~Morris.
\newblock Tarzan: A peer-to-peer anonymizing network layer.
\newblock In {\em Proceedings of the 9th ACM Conference on Computer and
  Communications Security}, CCS '02, pages 193--206, New York, NY, USA, 2002.
  ACM.

\bibitem{gennaro1998simplified}
R.~Gennaro, M.~Rabin, and T.~Rabin.
\newblock Simplified {VSS} and fast-track multiparty computations with
  applications to threshold cryptography.
\newblock In {\em Proceedings of the 17th Annual ACM Symposium on Principles of
  Distributed Computing}, PODC~'98, pages 101--111. ACM, 1998.

\bibitem{Juels04}
P.~Golle and A.~Juels.
\newblock Dining cryptographers revisited.
\newblock In C.~Cachin and J.~Camenisch, editors, {\em Advances in Cryptology
  -- EUROCRYPT '04}, volume 3027 of {\em Lecture Notes in Computer Science},
  pages 456--473. Springer Berlin Heidelberg, 2004.

\bibitem{Goodrich:2011:RSS:2049697.2049701}
M.~T. Goodrich.
\newblock Randomized shellsort: A simple data-oblivious sorting algorithm.
\newblock {\em J. ACM}, 58(6):27:1--27:26, Dec. 2011.

\bibitem{Goodrich:2012:POS:2133601.2133604}
M.~T. Goodrich, M.~Mitzenmacher, O.~Ohrimenko, and R.~Tamassia.
\newblock Practical oblivious storage.
\newblock In {\em Proceedings of the Second ACM Conference on Data and
  Application Security and Privacy}, CODASPY '12, pages 13--24, New York, NY,
  USA, 2012. ACM.

\bibitem{Hamada:2013:MPCSorting}
K.~Hamada, R.~Kikuchi, D.~Ikarashi, K.~Chida, and K.~Takahashi.
\newblock Practically efficient multi-party sorting protocols from comparison
  sort algorithms.
\newblock In T.~Kwon, M.-K. Lee, and D.~Kwon, editors, {\em Information
  Security and Cryptology -- ICISC 2012}, volume 7839 of {\em Lecture Notes in
  Computer Science}, pages 202--216. Springer Berlin Heidelberg, 2013.

\bibitem{Herzberg:1995:proactive}
A.~Herzberg, S.~Jarecki, H.~Krawczyk, and M.~Yung.
\newblock Proactive secret sharing or: How to cope with perpetual leakage.
\newblock In D.~Coppersmith, editor, {\em Advances in Cryptology -- CRYPTO
  '95}, volume 963 of {\em Lecture Notes in Computer Science}, pages 339--352.
  Springer Berlin Heidelberg, 1995.

\bibitem{Karlin:2013:BAS:2484239.2484290}
J.~Karlin, J.~Khoury, J.~Saia, and M.~Zamani.
\newblock Brief announcement: Scalable anonymous communication with {B}yzantine
  adversary.
\newblock In {\em Proceedings of the 2013 ACM Symposium on Principles of
  Distributed Computing}, PODC~'13, pages 128--130, New York, NY, USA, 2013.
  ACM.

\bibitem{asiacrypt-2010-23846}
A.~Kate, G.~M. Zaverucha, and I.~Goldberg.
\newblock Constant-size commitments to polynomials and their applications.
\newblock In {\em Advances in Cryptology -- ASIACRYPT 2010 -
  16\textsuperscript{th} International Conference on the Theory and Application
  of Cryptology and Information Security}, volume 6477 of {\em Lecture Notes in
  Computer Science}, pages 177--194. Springer, 2010.

\bibitem{Katz:2008:VSS}
J.~Katz, C.-Y. Koo, and R.~Kumaresan.
\newblock Improving the round complexity of vss in point-to-point networks.
\newblock In L.~Aceto, I.~Damg{\aa}rd, L.~Goldberg, M.~Halldorsson,
  A.~Ingolfsdottir, and I.~Walukiewicz, editors, {\em Automata, Languages and
  Programming}, volume 5126 of {\em Lecture Notes in Computer Science}, pages
  499--510. Springer Berlin Heidelberg, 2008.

\bibitem{ICDCN11}
V.~King, S.~Lonargan, J.~Saia, and A.~Trehan.
\newblock Load balanced scalable {B}yzantine agreement through quorum building
  with full information.
\newblock In M.~Aguilera, H.~Yu, N.~Vaidya, V.~Srinivasan, and R.~Choudhury,
  editors, {\em Distributed Computing and Networking}, volume 6522 of {\em
  Lecture Notes in Computer Science}, pages 203--214. Springer Berlin
  Heidelberg, 2011.

\bibitem{cyberwarfare2012}
A.~F. Krepinevich.
\newblock Cyber warfare: A nuclear option?, 2012.
\newblock Center for Strategic and Budgetary Assessments, Washington, DC, USA,
  2012.

\bibitem{Laur:2011:ObliviousDB}
S.~Laur, J.~Willemson, and B.~Zhang.
\newblock Round-efficient oblivious database manipulation.
\newblock In X.~Lai, J.~Zhou, and H.~Li, editors, {\em Information Security},
  volume 7001 of {\em Lecture Notes in Computer Science}, pages 262--277.
  Springer Berlin Heidelberg, 2011.

\bibitem{LeBlond:2013:TET:2486001.2486002}
S.~Le~Blond, D.~Choffnes, W.~Zhou, P.~Druschel, H.~Ballani, and P.~Francis.
\newblock Towards efficient traffic-analysis resistant anonymity networks.
\newblock In {\em Proceedings of the ACM SIGCOMM 2013 Conference on SIGCOMM},
  SIGCOMM '13, pages 303--314, New York, NY, USA, 2013. ACM.

\bibitem{Leighton:1990:Sorting}
T.~Leighton and C.~G. Plaxton.
\newblock A (fairly) simple circuit that (usually) sorts.
\newblock In {\em Proceedings of the 31st Annual Symposium on Foundations of
  Computer Science}, FOCS '90, pages 264--274, Oct 1990.

\bibitem{Machanavajjhala:2007:LDP:1217299.1217302}
A.~Machanavajjhala, D.~Kifer, J.~Gehrke, and M.~Venkitasubramaniam.
\newblock {$\ell$}-diversity: Privacy beyond {$k$}-anonymity.
\newblock {\em ACM Trans. on Knowledge Discovery from Data}, 1(1), Mar. 2007.

\bibitem{Nishide:PKC:2007}
T.~Nishide and K.~Ohta.
\newblock Multiparty computation for interval, equality, and comparison without
  bit-decomposition protocol.
\newblock In T.~Okamoto and X.~Wang, editors, {\em Public Key Cryptography --
  PKC 2007}, volume 4450 of {\em Lecture Notes in Computer Science}, pages
  343--360. Springer Berlin Heidelberg, 2007.

\bibitem{pfitzmann:86}
A.~Pfitzmann and M.~Waidner.
\newblock Networks without user observability -- design options.
\newblock In F.~Pichler, editor, {\em Advances in Cryptology -- EUROCRYPT '85},
  volume 219 of {\em Lecture Notes in Computer Science}, pages 245--253.
  Springer Berlin Heidelberg, 1986.

\bibitem{Reiter:1998:CAW:290163.290168}
M.~K. Reiter and A.~D. Rubin.
\newblock Crowds: Anonymity for web transactions.
\newblock {\em ACM Trans. Inf. Syst. Secur.}, 1(1):66--92, Nov. 1998.

\bibitem{shamir:how}
A.~Shamir.
\newblock How to share a secret.
\newblock {\em Commun. ACM}, 22(11):612--613, 1979.

\bibitem{vonAhn03}
L.~von Ahn, A.~Bortz, and N.~J. Hopper.
\newblock $k$-anonymous message transmission.
\newblock In {\em Proceedings of the 10th ACM Conference on Computer and
  Communications Security}, CCS '03, pages 122--130, New York, NY, USA, 2003.
  ACM.

\bibitem{Waidner:1990:DCD:111563.111630}
M.~Waidner and B.~Pfitzmann.
\newblock The dining cryptographers in the disco: Unconditional sender and
  recipient untraceability with computationally secure serviceability.
\newblock In J.-J. Quisquater and J.~Vandewalle, editors, {\em Advances in
  Cryptology -- EUROCRYPT '89}, volume 434 of {\em Lecture Notes in Computer
  Science}, page 690. Springer Berlin Heidelberg, 1990.

\bibitem{Yao:2005:Anonymity}
G.~Yao and D.~Feng.
\newblock A new k-anonymous message transmission protocol.
\newblock In C.~Lim and M.~Yung, editors, {\em Information Security
  Applications}, volume 3325 of {\em Lecture Notes in Computer Science}, pages
  388--399. Springer Berlin Heidelberg, 2005.

\bibitem{Zamani:2013:FOCI}
M.~Zamani, J.~Saia, M.~Movahedi, and J.~Khoury.
\newblock Towards provably-secure scalable anonymous broadcast.
\newblock In {\em the 3rd USENIX Workshop on Free and Open Communications on
  the Internet}, FOCI~'13, 2013.

\bibitem{Zhang:2011}
B.~Zhang.
\newblock Generic constant-round oblivious sorting algorithm for {MPC}.
\newblock In X.~Boyen and X.~Chen, editors, {\em Provable Security}, volume
  6980 of {\em Lecture Notes in Computer Science}, pages 240--256. Springer
  Berlin Heidelberg, 2011.

\end{thebibliography}
}}{\small \par}

\appendix

\section{Correctness and Security Proofs \label{sec:Security-Proofs}}

In this section, we prove the correctness and secrecy of our protocols.

\subsection{Proof of \textsf{Multiply}}

We have already showed the correctness of the \textsf{Multiply} and
now prove the secrecy of the algorithm. We first define $t$-secrecy,
a property required for the proof of secrecy.
\begin{defn}
\textbf{{[}$t$-secrecy{]}} A secret-shared value defined over a polynomial
$\phi$ is said to have \emph{$t$-secrecy} if and only if
\begin{enumerate}
\item it is \emph{$t$-private} meaning that no set of at most $t$ parties
can compute $\phi(0)$, and 
\item it is \emph{$t$-resilient} meaning that no set of $t$ or less parties
can prevent the other $n-t$ remaining parties from correctly reconstructing
$\phi(0)$.
\end{enumerate}
\end{defn}
\begin{lem}
\label{Lem:eVSSProps} The Shamir's secret sharing scheme has the
following properties.
\begin{enumerate}
\item A sharing defined over a polynomial of degree $d$ has $d$-secrecy
if the adversary has less than $d$ of the shares.
\item Let $\phi_{1}$ and $\phi_{2}$ be independent random polynomials
of degree $d_{1}$ and $d_{2}$ that correspond to sharings with $d_{1}$-secrecy
and $d_{2}$-secrecy respectively. $\phi_{3}=\phi_{1}+\phi_{2}$ is
a polynomial of degree $d_{3}=\max(d_{1},d_{2})$ that corresponds
to a sharing with $d_{3}$-secrecy.
\item Let $\phi_{1}$ and $\phi_{2}$ be polynomials of degree $d$ that
correspond to two sharings both with $d$-secrecy. $\phi_{3}=\phi_{1}\cdot\phi_{2}$
is a polynomial of degree $2d$ that corresponds to a sharing with
$d$-secrecy. %
\begin{comment}
Let $\phi_{1}$ be a polynomial of degree $d$ that defines a sharing
with $d$-secrecy. $\phi_{2}(x)=x\cdot\phi_{1}(x)$ is a polynomial
of degree $d+1$ that corresponds to sharing with $d$-secrecy.
\end{comment}

\end{enumerate}
\end{lem}
\begin{proof}
The first property follows from the definition. The second property
is correct due to the linearity Shamir's scheme. Without loss of generality,
let $d_{1}\leq d_{2}$. Intuitively, if we assume that the sharing
defined by $\phi_{3}$ does not have $d_{2}$-secrecy. Then, parties
compute $\phi_{2}=\phi_{3}-\phi_{1}$. Thus, they can find $\phi_{2}(0)$
(or similarly prevent others from learning $\phi_{2}(0)$). This contradicts
with the fact that $\phi_{2}$ corresponds to a sharing with $d_{2}$-secrecy.
For a complete proof, we refer the reader to Claim 3.4 of \cite{cryptoeprint:2011:136}.
The third property is correct because considering an arbitrary $P_{i}$
holding two shares $a_{i}$ and $b_{i}$ on polynomials $\phi_{1}$
and $\phi_{2}$ respectively, $P_{i}$ learns nothing from $a_{i}\cdot b_{i}$
other than what is revealed from $a_{i}$ and $b_{i}$ since $P_{i}$
computes it locally. So, the resulting shared value also has $d$-secrecy.\end{proof}
\begin{lem}
\label{lem:Multiply-Secrecy} Let $\la a\ra$ and $\la b\ra$ be two
secret-shared values both with $N/3$-secrecy and \textsf{\emph{$\la c\ra=\mathsf{Multiply(}\la a\ra,\la b\ra\mathsf{)}$}}.
The sharing $\la c\ra$ has $N/3$-secrecy.\end{lem}
\begin{proof}
In the first step of \textsf{Multiply}, algorithm\textsf{ GenRand}
creates two shared values $\langle u\rangle=(u_{1},...,u_{N})$ and
$\langle v\rangle=(v_{1},...,v_{N})$ that correspond to degree $N/6$
polynomials. Lemma~\ref{Lem:eVSSProps} proves that $\la u\ra$ and
$\la v\ra$ both have $N/6$-secrecy. For each party $P_{i},$ based
on Lemma~\ref{Lem:eVSSProps}, $w_{i}=u_{i}\cdot v_{i}$ defines
a new sharing that is on a polynomial of degree $N/3$ and has $N/6$-secrecy.
Moreover, $\langle a\rangle$ and $\langle b\rangle$ both are on
polynomials of degree $N/3$ and have $N/3$-secrecy. Thus, based
on Lemma~\ref{Lem:eVSSProps}, $c_{i}=w_{i}+\delta a_{i}+\varepsilon b_{i}-\varepsilon\delta$
defines a new sharing $\la c\ra$ that is on a polynomial of degree
$N/3$ and has $N/3$-secrecy. \end{proof}
\begin{lem}
\label{lem:Multiply} Given two secret-shared values $\la a\ra$ and
$\la b\ra$, the protocol \textsf{\emph{Multiply }}\emph{correctly
returns a shared value $\la c\ra$ such that $c=a\cdot b$.}\end{lem}
\begin{proof}
Follows from the proof of \cite{Beaver:1991}.
\end{proof}

\subsection{Proof of Theorem~\ref{thm:main}\textsf{ \label{sub:Security-of-Main}}}

In this section, we prove our main theorem (Theorem~\ref{thm:main}).
We prove in the real/ideal world model as described by Canetti~\cite{Canetti:UCSecurity:2001}.
First, we consider the protocol in an ideal model. In this model,
all parties send their input to a trusted party who computes the shuffling
circuit. Then, it sends the result to all parties. Let $A$ be the
sequence of inputs and $A^{\prime}$ be the sequence of sorted inputs
according to the random numbers associated with them. Recall that
we have at least $n-t$ honest parties. Based on Lemma~\ref{lem:permutation}
and conditioned on the event that no collision happens with high probability
(Lemma~\ref{lem:Collision}), the elements of $A^{\prime}$ that
correspond to honest parties can be any permutation of the elements
of honest parties in $A$. In other words, the probability that the
adversary can successfully map $A^{\prime}$ to $A$ is less than
$\frac{1}{(n-t)!}$ which guarantees $(n-t)$-anonymity (i.e., full
anonymity). Protocol~\ref{pro:Main} is the realization of the above
ideal model. The real model computes the circuit in a multi-party
setting. We prove this realization is correct and secure. 
\begin{description}
\item [{Setup.}] The correctness and secrecy follows from the proof of
Theorem~\ref{thm:quorum-building}.
\item [{Input~Sharing.}] The correctness and secrecy follows from the
proof of the VSS scheme used (\cite{Katz:2008:VSS,asiacrypt-2010-23846}).
After this step, each input quorum $Q_{i}$ has a correct sharing
of $P_{i}$'s input. This is the base case for our proof of circuit
computation step\textbf{.}
\item [{Random~Generation.}] The correctness and secrecy follows from
the proof of the \textsf{GenRand }algorithm.
\item [{Circuit~Computation.}] \emph{Correctness.} We prove by induction
in the real/ideal model. The invariant is that if the input shares
are correct, then the output of each gate is equal to the output when
the gate is computed by a trusted party in the ideal model, and the
result is shared between parties of the quorum correctly. For the
base case, note that the invariant is true for input gates. Induction
step is based on the universal composability of \textsf{$\mathsf{Compare}$}
and $\mathsf{Multiply}$. Moreover, based on the correctness of \textsf{$\mathsf{RenewShares}$},
the output resharing step generates new shares without changing the
output.\smallskip{}

\emph{Secrecy.} We prove by induction. The adversary cannot obtain
any information about the inputs and outputs during the computation
of each gate of the circuit. Let $Q$, and $Q^{\prime}$ be two quorums
involved in the computation of a gate, where $Q$ provides an input
to the gate, and $Q^{\prime}$ computes the gate. Consider a party
$P$. Let $S$ be the set of all shares $P$ receives during the protocol.
We consider two cases. First, if $P\notin(Q\cup Q^{\prime})$, then
elements of $S$ are independent of the shares $Q$ sends to $Q^{\prime}$.
Moreover, elements of $S$ are independent of the output of $Q^{\prime}$
since $Q^{\prime}$ also re-shares its output(s). Hence, $S$ reveals
nothing about the inputs and outputs of the gate.

Second, if $P\in(Q\cup Q^{\prime})$, then the inductive invariant
is that the collection of all shares held by dishonest parties in
$Q$ and $Q^{\prime}$ does not give the adversary any information
about the inputs and the outputs. As the base case, it is clear that
the invariant is valid for input gates. The induction step is as follows.
The adversary can obtain at most $2(N/6)=N/3$ shares of any shared
value during the computation step; $N/6$ from dishonest parties in
$Q$ and $N/6$ from dishonest parties in $Q^{\prime}$. By the secrecy
of the VSS scheme, at least $N/3+1$ shares are required for reconstructing
the secret. By the secrecy of $\mathsf{RenewShares}$ and $\mathsf{Multiply}$,
when at most $N/3$ of the shares are revealed, the secrecy of the
computation step is proved using universal computability of multi-party
protocols.
\begin{lem}
The gate computation step of Protocol~\ref{pro:Main} obliviously
swaps the pairs $(r,x)$ and $(r^{\prime},x^{\prime})$ according
to their first, i.e., if $r\leq r^{\prime}$, then this step outputs
$(\la s\ra,\la y\ra)=(\la r\ra,\la x\ra)$ and $(\la s^{\prime}\ra,\la y^{\prime}\ra)=(\la r^{\prime}\ra,\la x^{\prime}\ra)$.
Otherwise, it outputs $(\la s\ra,\la y\ra)=(\la r^{\prime}\ra,\la x^{\prime}\ra)$
and $(\la s^{\prime}\ra,\la y^{\prime}\ra)=(\la r\ra,\la x\ra)$.
In both cases, the adversary remains oblivious of which output is
mapped to the first (second) input pair.\vspace{-5pt}
\end{lem}
\begin{proof}
Let $\rho=\mathsf{Compare(}r,r^{\prime}\mathsf{)}$. Based on the
correctness and security of \textsf{Compare }(see \cite{Nishide:PKC:2007}),
$\rho=1$ if and only if $r\leq r^{\prime}$ and $\rho=0$, otherwise.
Also, the adversary does not learn anything about $r$ and $r^{\prime}$.
If $\rho=1$, then based on Equation~\ref{eq:switching}, the linearity
of the VSS scheme, and the correctness and secrecy of \textsf{Multiply
}(Lemmata~\ref{lem:Multiply-Secrecy} and \ref{lem:Multiply}), $s=r$,
$y=x$, $s^{\prime}=r^{\prime}$, and $y^{\prime}=x^{\prime}$. Otherwise,
$s=r^{\prime}$, $y=x^{\prime}$, $s^{\prime}=r$, and $y^{\prime}=x$. 
\end{proof}
\item [{Output~Propagation.}] The correctness and secrecy follows from
the proof of the subprotocol \textsf{Reconst}.
\end{description}

\subsubsection{Proof of Lemma~\ref{lem:protocol-costs} and Lemma~\ref{lem:protocol-costs-eVSS}
\label{sec:Proof-of-Costs}}

We first compute the cost of each step of the protocol separately
and then compute the total cost. Let $\nu_{1}(n)$ and $\nu_{2}(n)$
be the communication and computation complexity of the VSS subprotocol
respectively when it is invoked among $n$ parties. As stated in Theorem~\ref{thm:katz-vss}
and Theorem~\ref{thm:evss} both VSS protocols used in this paper
take constant rounds of communication.
\begin{itemize}
\item \emph{Setup.} The communication and computation costs are equal to
those costs of the quorum building algorithm of Theorem~\ref{thm:quorum-building},
which is $\tilde{O}(1)$ for each party. This protocol takes constant
rounds of communication.
\item \emph{Input~Broadcast.} The input broadcast step invokes the VSS
protocol $n$ times among $N=O(\log n)$ parties. So, this step sends
$O(n\cdot\nu_{1}(N))$ bits and performs $O(n\cdot\nu_{2}(N))$ operations.
Since the VSS scheme is constant-round, this step also takes constant
rounds.
\item \emph{Random~Generation.} It is easy to see that the subprotocol
\textsf{GenRand}\textsf{\textbf{ }}sends $O(N\cdot\nu_{1}(N))$ messages,
performs $O(N\cdot\nu_{2}(N))$, and has constant rounds.
\item \emph{Circuit~Computation.} The sorting network of Leighton and Plaxton~\cite{Leighton:1990:Sorting}
has \textsf{$O(n\log n)$} gates. So, the communication cost of this
step is equal to the communication and computation cost of running
\textsf{$O(n\log n)$ }instantiations of \textsf{Compare }and \textsf{RenewShares}.
\textsf{Compare} requires $O(\log q)$ invocation of\textsf{ Multiply}
which sends $O(N\cdot\nu_{1}(N))$ messages and computes $O(N\cdot\nu_{2}(N))$
operations. \textsf{RenewShares} also sends $O(N\cdot\nu_{1}(N))$
messages and computes $O(N\cdot\nu_{2}(N))$ operations. Hence, the
circuit computation phase sends $O(n\log n\cdot\log q\cdot N\cdot\nu_{1}(N))$
messages computes $O(n\log n\cdot\log q\cdot N\cdot\nu_{2}(N))$ .
Since the sorting network has depth $O(\log n)$, and \textsf{Compare}
takes constant rounds, this steps takes $O(\log n)$ rounds of communication.
\item \emph{Output~Propagation.} The costs are equal to the communication
and computation costs of running\textsf{ }$n$ invocations of \textsf{Reconst}
which costs $O(N^{2})$, plus sending the outputs to all parties,
which costs $O(n^{2}N)$. Thus, this step costs $O(n^{2}N)$. Since
\textsf{Reconst }is a constant-round protocol, this step takes constant
rounds.
\item \emph{Total.} Since $q>\frac{3}{2}kn^{2}\log n$, for a constant $k$,
$q=O(n^{3})$ and $\log q=O(\log n$). Using eVSS, we get $\nu_{1}(N)=\nu_{2}(N)=N^{2}=O(\log^{2}n)$.
Thus, Protocol~\ref{pro:Main} sends $O(n\log^{5}n)$ messages of
size $O(\kappa+\log n)$, computes $O(n(\kappa+\log n)\log^{5}n)$
operations for shuffling $n$ messages (excluding the output step).
This proves Lemma~\ref{lem:protocol-costs-eVSS}. For Lemma~\ref{lem:protocol-costs},
since $\nu_{1}(N)=\nu_{2}(N)=O(\mathsf{poly(}N\mathsf{)})$, Protocol~\ref{pro:Main}
sends $\tilde{O}(n)$ bits and computes $\tilde{O}(n)$ operations
for shuffling $n$ messages. This proves Lemma~\ref{lem:protocol-costs}.
In both cases, the output propagation step costs $O(n^{2}\log n)$
field elements. Finally, in both cases, the protocol requires $O(\log n)$
rounds of communication. This finishes the proof of Theorem~\ref{thm:main}.
\end{itemize}

\section{Remaining Algorithms\label{sec:Remaining-Algorithms}}

Beaver~\cite{Beaver:1991} describes a simple technique for generating
uniform random secrets by adding shared values chosen uniformly at
random by each party. Such a random shared value is used in several
parts of our protocol. The following subprotocol implements this technique.

{\small{}}
\begin{algorithm}[H]
{\small{\caption{\enskip{}\textsf{GenRand}}
\label{pro:GenRand}}}{\small \par}

{\small{Usage. Parties jointly generate a shared value $\la r\ra=(r_{1},...,r_{n})$
where $r$ is chosen uniformly at random from $\zp$.}}{\small \par}

{\small{\medskip{}
}}{\small \par}

\textsf{\uline{\small{$\mathsf{GenRand()}$}}}\textsf{\small{:}}{\small \par}

{\small{\smallskip{}
}}{\small \par}

{\small{For all $i\in[n]$, }}{\small \par}
\begin{enumerate}
\item {\small{Party $P_{i}$ chooses $\rho_{i}\in\zp$ uniformly at random
and secret shares it among all parties. }}{\small \par}
\item {\small{Let $\rho_{1i},...,\rho_{Ni}$ be the shares $P_{i}$ receives
from the previous step. $P_{i}$ computes $r_{i}=\sum_{j=1}^{n}\rho_{ji}$.}}\end{enumerate}
\end{algorithm}
{\small \par}

In the malicious setting, it is possible that dishonest parties send
spurious shares during secret reconstruction phase. In eVSS~\cite{asiacrypt-2010-23846}
(used in the cryptographic version of our protocol), this is solved
by asking all parties to broadcast a proof (called \emph{witness})
during reconstruction to verify broadcast shares. In our protocol,
reconstruction is postponed to after circuit computation. Since the
witnesses are generated in the sharing phase at the beginning of the
computation, and the witnesses do not have necessary homomorphic properties,
we cannot use them in our reconstruction phase. Instead, we correct
corruptions using a BCH decoding algorithm (e.g., the algorithm of
Berlekamp and Welch~\cite{Berlekamp:Welch:1986}) as in normal secret
reconstruction~\cite{Beaver:1991}. This technique is also used in
the VSS of Katz~et~al.~\cite{Katz:2008:VSS} and is implemented
in the following subprotocol.

{\small{}}
\begin{algorithm}[H]
{\small{\caption{\enskip{}\textsf{Reconst}}
\label{pro:Reconst}}}{\small \par}

\noindent {\small{Usage. Initially, all parties jointly hold a shared
value $\la a\ra=(a_{1},...,a_{n})$. Using this algorithm, parties
jointly reconstruct the secret, i.e., all parties learn the value
$a$.}}{\small \par}

{\small{\medskip{}
}}{\small \par}

\textsf{\uline{\small{$\mathsf{Reconst(\la a\ra)}$}}}\textsf{\small{:}}{\small \par}

{\small{\smallskip{}
}}{\small \par}

{\small{For all $i\in[n]$,}}{\small \par}
\begin{enumerate}
\item {\small{Party $P_{i}$ sends its share $a_{i}$ to all parties via
one-to-one communication.}}{\small \par}
\item {\small{Let $\alpha_{1},...,\alpha_{n}$ be the messages $P_{i}$
receives from the previous step. $P_{i}$ computes a degree $d=n/3$
polynomial $\phi(x)$ using the Lagrange interpolation polynomial,
\[
\phi(x)=\sum_{i=1}^{d+1}\alpha_{i}\prod_{j=1,j\neq i}^{d+1}(x-j)(i-j)^{-1}
\]
}}{\small \par}
\item {\small{For all $j\in[n]$, if there exists at least one $\alpha_{j}$
such that $\phi(j)\neq\alpha_{j}$, then $P_{i}$ runs the decoding
algorithm of Berlekamp and Welch\ \cite{Berlekamp:Welch:1986} to
recover the correct polynomial $\phi^{\prime}(x)$ of degree $d$.
For all $j\in[k]$, if $\phi^{\prime}(j)\neq\alpha_{j}$, then $P_{i}$
concludes that $P_{j}$ is dishonest and must be disqualified.}}\end{enumerate}
\end{algorithm}

\end{document}